\documentclass[12pt]{article}

\usepackage[top=1.25in, bottom=1.25in, left=1.5in, right=1.5in]{geometry}

\usepackage{fullpage}
\usepackage{times}
\usepackage{amsmath}
\usepackage[english]{babel}
\usepackage{amsthm}
\usepackage{bbm}
\usepackage{xcolor}
\usepackage{longtable}
\usepackage{booktabs, caption, makecell}

\usepackage{threeparttable}
\DeclareMathAlphabet{\mathpzc}{OT1}{pzc}{f}{it}

\usepackage{appendix}

\theoremstyle{plain}
\usepackage{graphicx}
\usepackage{ulem}
\usepackage{url}
\usepackage{xcolor}

\usepackage[noend]{algpseudocode}
\usepackage{algorithm}

\usepackage{comment}
\usepackage[autostyle,english=american]{csquotes}

\usepackage{setspace}

\usepackage{tabularx}
\usepackage{listings}

\usepackage{float}
\usepackage{parskip}
\usepackage{adjustbox}
\usepackage{lipsum}
\usepackage{tabularx}

\newtheorem*{myassum*}{The Hashrate Control Algorithm}
\newtheorem*{proposition*}{Proposition 1}
\newtheorem*{proposition**}{Proposition 2}
\newtheorem*{definition*}{Definition 1}
\newtheorem*{definition**}{Definition 2}

\usepackage{marvosym,soul,graphicx,
xcolor,endnotes,amssymb,amsmath,amsthm,setspace,soul}
\definecolor{rune}{HTML}{4A6672}

\title{Targeted Nakamoto:\\
A Bitcoin Protocol to Balance Network Security and Carbon Emissions} 

\providecommand{\keywords}[1]
{
  \small	
  \textbf{\textit{Keywords---}} #1
}

\author{Daniel Aronoff\thanks{Massachusetts Institute of Technology, Cambridge MA, USA}}

\date{August 21, 2025}

\begin{document}
\maketitle

\begingroup\renewcommand\thefootnote{}\footnotetext{\par \vskip 1em \noindent I wish to thank Madars Virza for his aid and advice in developing the core ideas of this paper; Neha Narula for helping to refine and clarify the presentation; thoughtful comments from participants at the Chaincode Labs research seminar and Mark Serrahn for producing the figures herein. Kristain Praizner wrote the code for the hashrate control algorithm and implemented it on an API that is a companion to this paper.}
\addtocounter{footnote}{0}\endgroup

\pagenumbering{roman}

\begin{abstract}

In a Proof-of-Work blockchain like Bitcoin, mining hashrate increases with block rewards. Higher hashrate reduces vulnerability to attacks (lowering network security cost) but raises carbon emissions and electricity use (raising environmental cost), reflecting a tradeoff with a hashrate point or interval that minimizes total cost. Targeted Nakamoto is a protocol that incentivizes miners to home hashrate in on the minimum cost range by capping block rewards when above target and imposing a block reward floor when below. Monetary neutrality is maintained by proportional adjustments in spending potential among UTXO holders, offsetting the additions and subtractions to the block reward caused by the protocol. Bitcoin faces two conflicting existential risks: currently, high mining energy consumption invites political backlash; in the future, reductions in miner rewards will cause hashrate to decline, lowering the cost of an attack. Targeted Nakamoto balances these concerns by guiding hashrate towards a chosen target that places an effective floor on the cost of attack and a ceiling on carbon emissions.

\end{abstract}
\keywords{Bitcoin, Blockchain, Mechanism Design, Network Security} 
\thispagestyle{empty}

\tableofcontents
\pagebreak

\newpage
\pagenumbering{arabic}
\onehalfspacing

\section{Introduction}
\label{sec: Introduction}

Proof-of -Work (PoW) blockchain cryptocurrencies like Bitcoin require the application of computing power by miners to operate the network. Miners assemble blocks and compete to solve a puzzle set by the code. The number of puzzle guesses (one hash per guess) a mining computer makes in a specified interval of time is its hashrate, which consumes electricity. Miners pay the direct cost of electricity. There are, in addition, two externalities created by mining activity. One externality is the vulnerability of the blockchain network to attack. The degree of vulnerability is related to the cost of forking the blockchain, which in turn is related to the cost of hashrate and the volume of hashrate employed by miners. The higher is hashrate, the higher is the cost of forking a PoW blockchain, which implies that network security is increasing in hashrate (equivalently, security cost is decreasing in hashrate). The other externality arises from the electricity consumed by hashrate. Miner demand for electricity increases both the price and aggregate consumption of electricity. The former is a cost imposed on competing users of electricity. The latter generates carbon emissions, which has a social cost. These electricity related costs are increasing in hashrate. The interaction of security risk and carbon emissions imply there is a tradeoff in total cost, defined as the sum of the two externalities costs, at different levels of hashrate. An increase in hashrate reduces the vulnerability of the network (equivalently, security cost is decreased), while it increases carbon emissions and creates upward pressure on the price of electricity. There may be a point or interval of hashrate in which total cost is minimized. Current PoW protocols, which I refer to as "Nakamoto", do not restrict hashrate.\footnote{  Our focus in this paper is Bitcoin, but the protocol applies to any PoW cryptocurrency.}

The key lever upon which the new protocol is built is the incentive effect of  the block reward. An increase in the USD value of the block reward creates an incentive for miners to apply more hashrate (assuming the dollar is the currency in which miners calculate profit). A decrease in the USD value of the block reward has the opposite effect. Equilibrium hashrate can range from zero to unbounded. Figure \ref{fig:Miner-rev-vs-energy-cons} shows the tight correlation over time between percentage change in the USD value of the Bitcoin block reward and the percentage change in energy consumption from Bitcoin mining.

\begin{figure}[h]
\centering
{\includegraphics[page=1,width=0.65\textwidth,height = 0.35\textheight]{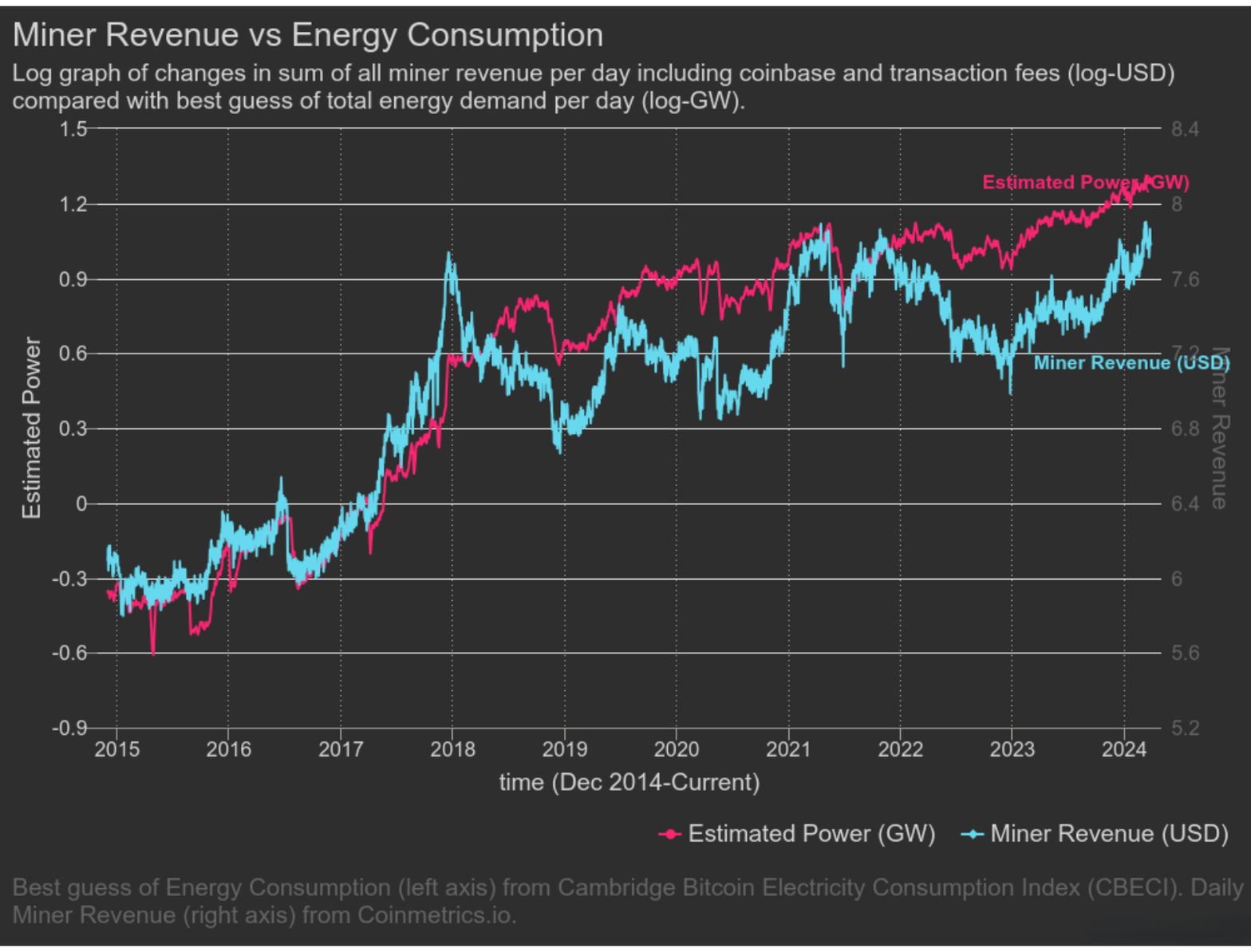}}
\caption{Bitcoin Energy Consumption vs USD Mining Rewards. Source: CBECI\protect\footnotemark and Coinmetrics.\protect\footnotemark}
\label{fig:Miner-rev-vs-energy-cons}
\end{figure}
\footnotetext {\url{https://ccaf.io/cbnsi/cbeci}}\footnotetext{\url{https://coinmetrics.io}}

 Targeted Nakamoto (hereafter "\textit{Targeted}") modifies Nakamoto to create an incentive for miners to push hashrate toward an interval whenever it strays outside the interval. The miner's minimum block reward (in BTC) is increased when hashrate is below the interval, which induces miners to apply more hashrate. The miner's maximum block reward (in BTC) is decreased when hashrate is above the interval, which induces miners to apply less hashrate. Monetary neutrality is maintained by offsetting the increase and decrease in the miner's block reward by opposite adjustments to the spending potential of addresses with UTXO value, in proportion to the value held by each address. 
 
 \subsection{Balancing Competing Risks}

 \textit{Targeted} balances the ecological threat from electricity consumed by hashrate against the threat to network security from low hashrate. The former is a subject of current political controversy, while the latter will materialize in the future when the USD value of the block reward declines as a result of the halving of the coinbase. \footnote{Historically, the trend USD/BTC exchange rate has risen sufficiently to offset the halving of the coinbase. However, this cannot continue indefinitely for two reasons. First, coinbase will eventually be zero. Second, maintaining a fixed USD value of the coinbase paid to miners as halving continues would require the inflation adjusted USD to BTC exchange rate to double every four years. This implies a growth rate that exceeds US (and world) GDP by orders of magnitude, which is dynamically unsustainable (Tirole \cite{Tirole}).} Transactors will have little incentive to compensate by increasing fees, which are paid to obtain a priority in the qeue of mempool transactions. The decline in block reward will elicit a decline in hashrate (Equation \ref{eq:Mining profit}). \textit{Targeted} compensates by increasing the block reward, and does it without increasing overall UTXO spending potential.

\subsection{Links to code and web API (by Kristian Praizner)}

Attached are links to the code-base for controlling Hashrate and a Web API with an empirically estimated hashrate function and a module that enables the user to choose a target hashrate interval and control parameters and to run counterfactual experiments with alternative time-paths of model variables. 

The code-base is here \url{https://github.com/Krisp140/TargetedNakomoto}

The web API is here \url{https://targetednakamoto.com/About}.

\subsection{Related Literature} 

This paper builds on the contributions of Nakamoto \cite{BitcoinWhitePaper}  who laid the foundations for the PoW protocol used in Bitcoin. 

There are two previous works with which aspects of this study closely align. Pagnotta \cite{Pagnotta2021} identified and modeled the complimentarity between mining hashrate and Bitcoin USD price. The interdependence runs through network security. Hashrate determines security, security affects traders valuation of Bitcoin, which affects it price, price affects the USD value of the mining reward, which determines equilibrium hashrate. \textit{Targeted} injects into this causal chain a regulator that adjusts the miner's block reward to incentivize miners to keep hashrate (and thereby security) in the neighborhood of a selected target. Mirkin et.al. \cite{mirkinsprints} show that electricity cost can be reduced, for any level of network security, by adding a ''proof-of-delay'' (e.g. using a verifiable delay function) threshold for a miner's block to be eligible to append to the chain. The mandatory delay interval creates an incentive for miners to shift expenditure from hashrate to investment in mining equipment. The re-allocation of expenditure from variable to fixed cost reduces the carbon footprint. Notably, proof-of-delay does not alter equilibrium expenditure on mining, which is determined by the miner's block reward. This means it does not affect the cost to attack the network and thereby network security. While proof-of-delay reduces energy consumption for a given level of security, \textit{Targeted}, which is compatible with proof-of-delay, enables network security to be selected and maintained.

One of the externalities costs addressed in this paper are the carbon emissions generated by the electricity consumed in mining Bitcoin. I rely on two data providers. One is the Bitcoin Electricity Consumption Index  (BECI) of the Cambridge U.K. Centre for Alternative Finance \cite{CambridgeBitcoinElectricityConsumptionIndex}, which is regarded as a credible - though not infallible - source for information on Bitcoin energy consumption. The other is Coinmeterics \footnote{\url{https://coinmetrics.io}}(which is a source for BECI) for other Bitcoin performance indicators. Coinmetrics is considered to be a careful aggregator of data on cryptocurrencies. Sedlmeir et.al. \cite{Sedlmeir} and Gallersdorfer et.al. \cite{Gallersdorfer} review comparative magnitudes of energy consumption for different blockchain protocols. Gallersdorfer \cite{Gallersdorfer} compares the estimation of Bitcoin electricity consumption across several sources. For my purposes, the salient finding is that these studies document a time-path of energy consumption that is tightly correlated with the time-path of hashrate in Bitcoin. \footnote{See Figure 1 in Gellersdorfer \cite{Gallersdorfer}.} 

As with several previous studies of blockchain protocol, this paper uses mechanism design approach to the study of the consensus protocol.  Chen et.al. \cite{CHEN} surveys various aspects of existing PoW blockchain protocols, including the incentive to increase energy consumption in mining,  from a mechanism design perspective. Cong et.al \cite{Cong} analyses the effect of PoW mining competition on energy consumption.  Leshno and Strack \cite{Leshno} identifies features of the miner payoff structure that are required to maintain blockchain stability.  They prove that certain of the core properties of a decentralized PoW blockchain - namely free entry of anonymous profit driven miners - are only incentive compatible with a reward mechanism that is similar to Nakamoto. \textit{Targeted} preserves those features. 

In their study of the Bitcoin payment system, Huberman et.al. \cite{HubermanLeshno} characterize an equilibrium where transactions are  appended to blocks in descending order of the size of fee offered by transaction senders, which reflects transactor's relative willingness to pay for a position in the queue for placement on the blockchain. The adjustment to the miner's block reward in \textit{Targeted} is guided by this insight. When the block reward is adjusted, the  ordering of the rewards received by the miner from transactions always matches the order of fees offered by transaction senders. This, in turn, leaves unaltered he order in which transactions are appended to blocks.

\subsection{The research question and design constraints} 

\textit{Targeted} augments Nakamoto by creating an incentive for miners to home in on a hashrate interval. The augmentation comes close to adhering to formal constraints which  reflect core values held by the Bitcoin community.\footnote{What constitutes a core feature is subjective. I developed the list herein after consultation with a number of participants in the Bitcoin network, including developers, miners and exchange operators.} The constraints are partitioned into the consensus protocol layer (\text{C} constraints) and the transaction protocol layer (\textbf{T}- constraints).

$\text{\textbf{C}}\; \text{constraints}\left\{
\begin{tabular}{p{.8\textwidth}}
\begin{itemize}

\item No change to proof-of-work mining 

\item No change to growth rate target and puzzle difficulty adjustment policy

\item No change in block size

\item No change to the monetary policy (aggregate UTXO value unaffected)

\item No required oracle

\item No hard forks

\end{itemize}
\end{tabular}
\right.$

$\text{\textbf{T}}\;\text{constraints}\left\{
\begin{tabular}{p{.8\textwidth}}
\begin{itemize}

\item No interference with transactions between nodes

\item No change to the order in which transactions are appended to the blockchain
\end{itemize}
\end{tabular}
\right.$

It will be shown that \textit{Targeted} satisfies each of these constraints with one exception. \textit{Targeted} can instruct miners to burn UTXO's from transaction outputs, which may require a hard-fork to implement.

\subsection{Roadmap} 

The rest of the paper is organized as follows. Section \ref{sec:Hashrate Externalities} provides an overview of the two types of cost, the sum of which \textit{Targeted} is designed to minimize; environmental costs that increase with hashrate and network sercurity costs that decrease with hashrate. Section \ref{sec: Targeted Nakamoto - A Mechanism Design Perspective} casts Nakamoto and \textit{Targeted} in a mechanism design framework in order to articulate the key features of each protocol and their interaction and to show how \textit{Targeted} augments Nakamoto. Section \ref{sec: The Model} shows how the ratio of mining puzzle difficulty to blockchain growth rate is a sufficient statistic for hashrate and then derives the expression for equilibrium hashrate as a function of, inter alia, the value of the block reward. Section \ref{sec: The Targeted Block Reward Policy} uses the relationship between the puzzle difficulty/growth rate ratio and hashrate, along with hashrate equilibrium to construct the policy for adjusting the miner's block reward to incentivize miners to push hashrate toward the target hashrate interval. Section \ref{sec:Strategic Effects of Block Adjustment Policy} demonstrates that the block reward policy is consistent with Nash equilibrium at the targeted hashrate and demonstrates that the hashrate control algorithm dynamically pushes hashrate toward the target interval. Section \ref{sec: Monetary Neutrality} states a protocol that renders the aggregate money supply and the distribution of spending potential unaffected by the Hashrate Control Algorithm. The paper concludes with a summary of the key properties of \textit{Targeted}.

\section{Hashrate Externalities}
\label{sec:Hashrate Externalities}

 There are intensive and extensive margins related to hashrate in Bitcoin mining. The intensive margin is the equilibrium hashrate cost paid by miners. The extensive margins are the social cost of $CO^{2}$ emissions from the electricity used to generate hashrate and the network security risks arising from the cost of an attack that forks the blockchain.  The former is increasing and the latter is decreasing in aggregate miner hashrate. The causal link from hashrate to these variables is determined by ASIC efficiency and electricity cost. Equation \ref{eq:unit_N_cost} displays the relationship.

\vspace{-5pt}
\begin{equation}
\label{eq:unit_N_cost}
\underset{\text{per-hash cost}}{\underbrace{cost/N}} = \underset{\text{electricity cost}}{\underbrace{(cost/KwH)}}\div\underset{\text{ASIC efficiency}}{\underbrace{(N/KwH)}}\end{equation}

Where $N$ is hashrate (defined in Section \ref{sec: The Model}), $KwH$ is kilowatt hours. Environmental cost is an increasing function of hashrate as follows.

\vspace{-5pt}
\[\text{environmental cost} \underset{\text{external function}}{\underbrace{\longleftarrow}} CO^{2} \underset{\text{physical law}}{\underbrace{\longleftarrow}} KwH = \textbf{\textit{N}}\div \text{ASIC efficiency}\]

Network security cost is a decreasing  function of hashrate as follows.

\vspace{-5pt}
\[\text{network security cost}\underset{\text{external function}}{\underbrace{\longleftarrow}}\text{mining cost} = \textbf{\textit{N}}\times \text{electricity cost}\div \text{ASIC efficiency}\]

Figure \ref{fig: PoW Cost Graph 1} displays the relationship between environmental cost and network security cost on the vertical axes, and hashrate on the horizontal axis.  is a upward sloped curve and network security is a downward sloped curve. We define optimal hashrate interval as the point, or interval, where the sum of the costs are minimized. The goal of \textit{Targeted} is to incentivize miners to home in on this interval. 

\vspace{-5pt}
\begin{figure}[h]
\centering
\includegraphics[page=1,width=0.55\textwidth,height = 0.27
\textheight]{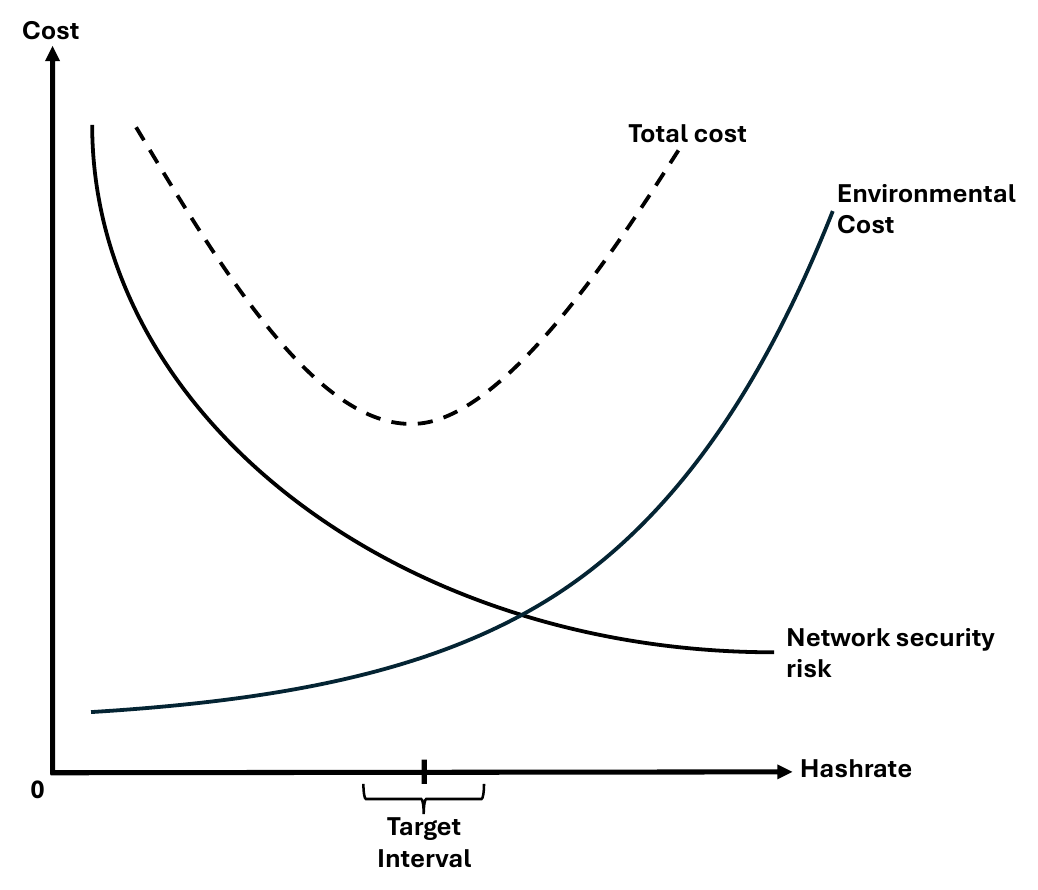}
\caption{Target Hashrate Interval}
\label{fig: PoW Cost Graph 1}
\end{figure}

\subsection{Dependencies}

Assuming the external functions are stable over time, environmental risk and network security are conditional on ASIC efficiency and electricity cost.\footnote{The function connecting electricity to $CO^{2}$ is a physical law. The function connecting forking cost to attack incentives is subject to variables that lie outside our analysis, such as, inter alia, interest rates and the legal environment.}  An increase in ASIC efficiency will shift both curves rightward and increase the optimal hashrate target. An increase in electricity cost will shift the security curve rightward while leaving the environmental risk curve unaffected, and reduce the optimal hashrate. ASIC efficiency and electricity cost can be estimated and oracles can be implemented to adjust the hashrate target. Additionally, there is a causal channel running from hashrate to electricity cost, whereby the increased demand for electricity places upward pressure on the price. Fowlie et.al. \cite{Fowlie} estimate a price elasticity of demand for electricity in the U.S. of -2.06 with a range of alternative elasticities around -1, which implies that an increase in electricity demand will be split between increases in both price and quantity. 

Less clear is how to measure the external functions. There is a literature on measuring the social cost of $CO^{2}$ emissions. Nordhaus \cite{Nordhaus}), Deitz et.al. and Stern \cite{Stern} are important contributions and Chen et.al. \cite{CHEN} model the flow of carbon emissions from multiple energy sources. Nevertheless, the shape and magnitude of the relationship between $CO^{2}$ emissions and social cost is a matter of ongoing debate. There has not been any substantive work, to our knowledge, connecting mining cost to quantifiable network security.\footnote{Appendix \ref{app:network_security} discusses the relationship between mining cost and network attack vectors.} A market in Bitcoin security risk, with payoffs tied to attack events, could resolve the matter. But for now, this function can only be guessed at. The practical implementation of \textit{Targeted} will require the formation of a consensus among Bitcoin network participants on the external functions. This may require negotiations that will involve compromise over differing opinions and economic interests.

\subsection{Network Security}
\label{app:network_security}

To elucidate the relationship between hashrate cost and network security, we review three vectors of attack on a PoW blockchain that decrease in profitability as hashrate cost increases. 

\textit{Double-spend attack} A double-spend attack requires the attacker to build an alternative chain that overtakes the incumbent canonical chain at in terms of cumulative mining puzzle difficulty. An attacker with over 50\% of total hashrate will overtake the incumbent canonical chain at some block with probability measure 1. At less than 50\% share, the attacker's probability of success is less than half.\footnote{See Gervais et.al. \cite{Gervais} for simulations of attacks using different ratios of attacker hashrate to miner hashrate. The frequency of success is increasing in attacker hashrate and the average number of blocks required to overtake the incumbent canonical chain is decreasing in attacker hashrate.} The amount of attacker hashrate required to successfully fork the blockchain, and its cost, is increasing in the level of honest miner hashrate. Budish  \cite{Budish} and Aronoff et.al. \cite{ADESS} point out that the cost of a successful attack is at least partially offset by the block rewards that accrue to the attacker. The possibility that a victim of a double-spend attack retaliates by mining on the incumbent canonical chain leads to a war of attrition between the attacker and the victim (Moroz et.al. \cite{moroz2020doublespend}). Aronoff et.al. \cite{ADESS} show that the outcome of a war of attrition depends on equilibrium refinements and exigent circumstances, even when the costs between players differ. Nevertheless, an increase in the cost of launching the attack will act as a deterrent by increasing the sunk cost that the attacker loses if it is unsuccessful. Equation \ref{eq: Expected Profit from Double-Spend Attack} is the attacker's un-discounted\looseness=-1 ex-ante expected profit when block rewards equal hashrate cost. The key observation is that, for any probability of success, the attacker's expected profit is declining in honest miner hashrate.
\vspace{-5pt}
\begin{multline}
\label{eq: Expected Profit from Double-Spend Attack}
\text{Ex-ante Expected Profit}\;=\; Prob(success)V - \\
\{1 - Prob(success)\}\text{E[hashrate cost per block]E[number of blocks]}
\end{multline}

Where $V$ is the double-spend value. 

\textit{Eclipse attack} An eclipse attack involves isolating a node from the network and constructing a chain, forked from the incumbent canonical chain, that is broadcast only to the victim (the "false chain"). On the false chain the attacker sends the victim UTXO value, which the victim believes to be on the canonical chain. After the attacker receives the exchange item from the victim, it discontinues applying hashrate to the false chain. The attacker incurs two elements of cost. One element is the hashrate the attacker applies to the false chain. The other element is the cost of isolating the victim from the network. Neither of these costs necessarily increase with honest miner hashrate. However, the attacker may need to apply hashrate that is proximate to honest miner hashrate to convince the victim into believing that the fake chain is canonical. Therefore it is possible that the cost of an eclipse attack is increasing in honest miner hashrate, in which event the profit is declining in honest miner hashrate.

\textit{Shorting attack} A shorting attack has two parts; (i) placing a bet that a cryptocurrency will decline in value relative to some other currency or asset and (ii) gaining control of the cryptocurrency network in order to cause a disruption that will induce a loss in confidence and fall in exchange value. One type of bet is a forward sale of, say, Bitcoin in dollars. The attacker contracts to sell BTC to a counterparty at a future date for a dollar strike price per unit of BTC. The attacker anticipates that its disruption to Bitcoin will enable it to purchase BTC at a dollar price that is below the strike price, enabling it to resell to the counterparty at a profit. The size of the bet is independent of hashrate. In a liquid futures market an attacker could contract to purchase \$x dollars worth of BTC at some future price regardless of the level of hashrate. On the other hand, the cost of gaining control of the network scales with hashrate. Thus, the profitability of the attack is declining in hashrate.\footnote{The value of the block rewards earned by the attacker decline with the Bitcoin exchange value.}

Finally, Carlesten et.al. \cite{BitcoinwotheBlockReward} argued that as the coinbase approaches zero and miners become dependent on transaction fees, the decline and variability of the fees will incentivize selfish mining, which could disrupt the liveness of the blockchain and possibly push down hashrate (and thereby mining cost).

\section{Targeted Nakamoto - A Mechanism Design Perspective}
\label{sec: Targeted Nakamoto - A Mechanism Design Perspective} 

This section describes the elements of \textit{Targeted} and compares the mechanism designs of \textit{Targeted} and Nakamoto. 

\subsection{Key Building Blocks of Targeted Nakamoto}

\textit{Targeted} controls hashrate by leveraging four features that are intrinsic to the blockchain. One feature is informational. The hashrate, which is not recorded on the blockchain, can be estimated from mining puzzle difficulty and block timestamps, which are recorded on the blockchain. The augmented protocol targets bounds on hashrate, for which the ratio of puzzle difficulty to blockchain growth rate is a sufficient statistic.

The second feature is an object that can be modified. The Nakamoto block reward is comprised of two parts; an algorithmic minting of new tokens and a discretionary fee added to a transaction. Both parts of the block reward  are recorded in the blockchain and can be altered by the code. \textit{Targeted} places a floor or ceiling on the the miner's block reward when hashrate is above or below the target interval.

The third feature is a market incentive. The hashrate applied by miners is affected by the block reward; it increases as the value of the block reward increases and decreases as the value of the block reward decreases. This implies that hashrate can be controlled by manipulation of the block reward.

The fourth feature enforces monetary neutrality by adjusting the aggregate UTXO spending potential to exactly offset the amount added to or subtracted from the total block reward to obtain the miner's block reward.

\subsection{An overview of the mechanism design of the protocols}

\textit{Targeted} can be decomposed into two protocols that interact while operating separately. The Nakamoto part is the mechanism associated with the current protocol. The \textit{Targeted} part adds a mechanism to achieve a target hashrate interval. Figure \ref{fig: Blockchain Records} displays the information recorded on the blockchain that is used by each protocol\footnote{The network also records the addresses of UTXO'a at each block. Neither Nakamoto or \textit{Targeted} use that information.}. 

\begin{figure}[tb]
\centering
\includegraphics[page=1,width=0.7\textwidth,height = 0.16
\textheight]{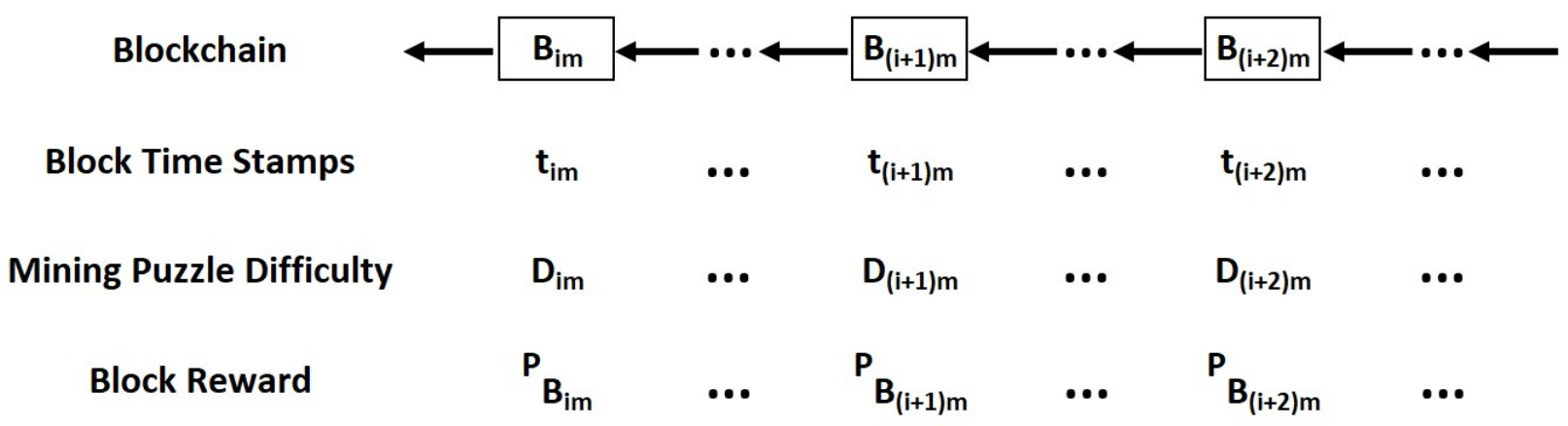}
\caption{Blockchain Records}\label{fig: Blockchain Records}
\end{figure}

Table \ref{table: A Comparison of Mechanisms} shows the basic elements of Nakamoto and \textit{Targeted}. For each protocol, the information contained in the signal, the instrument and the feedback are recorded on the blockchain. There is a salient difference in the target. The target for Nakamoto, blockchain growth rate, is recorded on the blockchain. The target for \textit{Targeted}, hashrate, is not recorded on the blockchain.

\begin{table}[H]
\centering
\begin{tabular}{ |p{1.9cm}||p{4.9cm}|p{4.9cm}|  }
 \hline
 \multicolumn{3}{|c|}{Protocol Mechanism Design} \\
 \hline
 Element & Nakamoto & \textit{Targeted}\\
 \hline
  Target & Blockchain Growth Rate & Hashrate Target Interval\\
  Signal & Block Timestamps & Puzzle Difficulty/Growth Rate\\
  Feedback & Block Timestamps & Puzzle Difficulty/Growth Rate\\
 Instrument & Puzzle Difficulty & Block Reward\\
 Policy & Adjust to Growth Rate & Ceiling or Floor on Block Reward\\
 \hline
\end{tabular}
\caption{A Comparison of Mechanisms}\label{table: A Comparison of Mechanisms}
\end{table}

\section{The Model}
\label{sec: The Model}

This section derives and characterizes the equilibrium level of hashrate. Section \ref{subsec:The puzzle difficulty signal of hashrate} derives the function of puzzle difficulty that is a sufficient statistic for hashrate. Section \ref{subsec: Equilibrium} characterizes the equilibrium hashrate as a function of, inter alia, the miner's block reward.

\subsection{The puzzle difficulty signal of hashrate}
\label{subsec:The puzzle difficulty signal of hashrate}

Nakamoto sets a target time for intervals between the addition of new blocks to the blockchain. The timestamps are read by the code at intervals of $m$ blocks, which is an epoch. An epoch in Bitcoin  $m = 1,260$ blocks. The instrument used to achieve the target is an adjustment to the difficulty of the mining puzzle, denoted as $D$. The mining puzzle is the solved when the miner's input into a hash function outputs a value that falls within the target range set by the code. Puzzle difficulty is the ratio of the range set by the code to the range of all possible output values. The probability of guessing the solution is $p = 1/D$ and each guess is independent.\footnote{By construction $p$ is a monotone decreasing function $P(D)$ of $D$. The derivative $\frac{d p}{d D} <0$ in both the simplified expression in the text as well as e.g. the algorithm used in Bitcoin Core. The derivative is the only property material to our analysis.} $D$ is adjusted after the $m$th (final) block of a epoch by dividing the number of blocks into the elapsed time from the previous timestamp. When the blockchain growth rate exceeds the target, puzzle difficulty is increased. When the blockchain growth rate falls below target, puzzle difficulty is decreased. 

A single mining computer can guess a puzzle solution sequentially in time.\footnote{A computer with $n$ parallel processors is treated here as $n$ separate computers.} I model mining computers as making successive guesses synchronously at times $t\in \{t_{1},t_{2},...\}$, each separated by an interval of temporal length $ t_{i+1} - t_{i}$ which I normalize to 1 unit of time\footnote{The synchronicity enables the expression of the time-path of puzzle solutions compactly as a Bernoulli Process, which simplifies the analysis WLOG.}. The Nakamoto blockchain growth target is one block every $T^{*}$ units of time. During this time interval each computer has guessed (or hashed) $T^{*}$ proposed solutions. I define a unit of hashrate as a single mining computer making $T^{*}$ puzzle guesses. If there are $N$ mining computers making guesses, the number of total guesses made in the time interval of the blockchain growth target is $NT^{*}$ and the hashrate is $N$.\footnote{ This accords with the conventional definition of hashrate used e.g. by Coinmetrics \cite{CoinmetricsCMBI}, as puzzle guesses per second.} The framework for guessing puzzle solutions over time is a Bernoulli Process. Let $T$ denote the time the first puzzle solution is broadcast to the network. For a hashrate of $N$ the expected time to reach the solution is expressed by Equation \ref{eq: Expected Time to Solve Mining Puzzle}.

\begin{equation}
\label{eq: Expected Time to Solve Mining Puzzle}
\underset{\text{expected time between blocks}}{\underbrace{\mathbb{E}[T]}} = D/N \footnote{The formula for the expected time of first solution is: $\mathbb{E}$[Number of guesses] = 1/Prob, where Pr is the probability of solving the puzzle. In our case, $Prob = 1/D$ and Number of guesses = $NT$. When $N$ is a fixed number, the Bernoulli Process implies that $N\mathbb{E}[T]= 1/Pr = D$. This yields the expression in Equation \ref{eq: Expected Time to Solve Mining Puzzle} }
\end{equation}

Equation \ref{eq: Expected Time to Solve Mining Puzzle} shows that, in order to achieve the growth rate target of $T^{*}$ in expectation, puzzle difficulty $D$ must be adjusted to offset hashrate $N$. Since $T$ for the prior epoch is recorded on the blockchain and $D$ is set by the code, the formula can be inverted to estimate hashrate as a function of the growth rate and puzzle difficulty, $\hat{N} = D/T$.\footnote{ $\hat{N} = N + \epsilon$, where $\epsilon$ is a mean zero iid random variable.} $\hat{N}$ is unbiased since the observed $T$ is an unbiased estimate of the mean time interval implied by the Bernoulli Process. It is also a sufficient statistic for hashrate since $D$ and $T$ contain all of the information on $N$ that is encoded on the blockchain.

\subsection{Mining equilibrium}
\label{subsec: Equilibrium}

The miner who is first to broadcast a puzzle solution receives a block reward of cryptocurrency. Miners incur cost in terms of the local official currency, which we denote USD.\footnote{Profit can be expressed in terms of cryptocurrency. In that case the unit hashrate cost cost $c$ is multiplied by $1/e$, which is the USD to BTC exchange rate. The equilibrium hashrate is the same  whatever currency is chosen to be the unit of account.} Mining equilibrium is a state where miners choose hashrate to maximize expected profit. When miners adjust hashrate at short block intervals profit maximization implies that the profit from a marginal unit of hashrate equals the competitive rate of profit. The independence of puzzle guesses means each unit of hashrate has a $1/N$ probability of receiving the block reward. Equation \ref{eq:Mining profit} is the equilibrium profit from mining one unit of hashrate (i.e. operating one mining computer for $T^{*}$ units of time).\footnote{Equation \ref{eq:Mining profit} is a dynamic equilibrium in the sense that it includes a growth rate that deviates from target. A rigorous expression would show $\Delta_{T}$ as a function of puzzle difficulty and hashrate.}
\begin{equation}
\label{eq:Mining profit}
\pi = \frac{1}{N}eP_{B_{n}}^{M}(1 + \Delta_{T}) - c = \phi
\end{equation}

$\Delta_{T} =  T^{*} - T$ is the difference between the target growth rate, $T^{*}$, and the expected actual growth rate $T$.\footnote{Blockchain growth rate is expressed as an expectation since it is stochastic.} $(1 + \Delta_{T})$ is the number of blocks per $T^{*}$ units of time; $c$ is the marginal cost of a unit of hashrate (i.e. one computer for $T^{*}$ units of time);\footnote{Marginal cost is the cost of electricity to run the ASIC's. Fixed costs, such as the rent on the building where the computers are stored, or the sunk cost of purchasing the computers, are not included.} $e$ is the USD/BTC exchange rate and $p_{B_{n}}^{M}$ is the expected miner's block reward on the $n^{th}$ block. ($p_{B_{n}}$ is the total block reward). I make no assumption about the structure of the mining market. $\phi \geq 0$ a competition parameter. Under free entry in the mining market $\phi =0$. Under less competitive conditions $\phi > 0$. Equation \ref{eq:Profit bound - static} is the mining market equilibrium. 

\begin{equation}
\label{eq:Profit bound - static}
N = eP_{B_{n}}^{M}(1 + \Delta_{T})/(c + \phi ) 
\end{equation} 

Equation \ref{eq:Profit bound - static} shows that hashrate $N$ is a function of the market conditions faced by miners. Puzzle difficulty does not appear in Equation \ref{eq:Profit bound - static}. Equation \ref{eq: Link between the Signal and the Instrument} substitutes $D/T$ for $N$ in Equation \ref{eq:Profit bound - static}.

\begin{equation}
\label{eq: Link between the Signal and the Instrument}
D/T = eP_{B_{n}}^{M}(1 + \Delta_{T})/ (c + \phi) - \epsilon\protect\footnotemark
\end{equation}
\footnotetext{$\epsilon$ is the stochastic noise from the relation $\hat{N} = N + \epsilon$.}

Equation \ref{eq: Link between the Signal and the Instrument} displays the monotone relationship between the miner's block reward, $p^{M}_{B}$, and the sufficient statistic for hashrate, $D/T$. 


\section{The \textit{Targeted} Block Reward Policy}
\label{sec: The Targeted Block Reward Policy}

This Section explains the implementation of the policy of adjusting the miner's hashrate. Section \ref{subsec: The block reward adjustment allocations} shows how the allocation of an addition or subtraction to the miner's block reward leaves unchanged the miner's incentive to append transactions onto blocks according to the ranking of the size of fee offered by the transactor. Section \ref{subsec: The mechanism} builds on the analysis of Section \ref{sec: The Model} to state the policy of adjusting the miner's block reward in response to the ratio of puzzle difficulty to blockchain growth rate. Section \ref{subsec: Policy to control hashrate} states the algorithm for controlling the miner's block reward.

\subsection{The block reward adjustment allocations}
\label{subsec: The block reward adjustment allocations}

One of the \textbf{T} constraints is that $Targeted$ should not change the order in which transactions are appended to blocks. Huberman et.al. \cite{HubermanLeshno} demonstrate that Nakamoto induces an equilibrium ordering of transaction fees and transactions on the blockchain that reflects transactors' relative time preferences. Transactions are appended to blocks in descending order of the size fee offered and the ranking of fees match the ranking of transactors' time-preference for executing their transaction.\footnote{See Theorem 1 and Proposition 2. Note that the result relies on parametric assumptions. The reader can consult Huberman et.al. \cite{HubermanLeshno} for details.} It is a constrained efficient outcome in the sense that if agents A and B send their respective transactions to the mempool at the same time and agent A incurs a higher cost for delay in completing its transaction than agent B, agent A's transaction will be appended to the blockchain ahead of agent B's transaction. \textit{Targeted} preserves the same ordering of transactions by adjusting the coinbase and each transaction fee by the same percentage. $Fee_{f,B_{n}}$ denotes the fee paid by the $f$th transaction appended to block $B_{n}$. Denoting $\xi = (p^{M}_{B_{n}} - p_{B_{n}})/p_{B_{n}}$, the ordering $Fee_{i, B_{n}} > Fee_{j,B_{n}}$ implies the ordering $\xi Fee_{i,B_{n}} > \xi Fee_{j,B_{n}}$, i.e. the ordering of transaction fees sent to the miner is unchanged. $R_{n}$ denotes the coinbase portion of the block reward. The miner's block reward is $p^{M}_{B_{n}} = (1 + \xi)\big[\sum_{f}Fee_{f,B_{n}} + R_{n}\big]$. When $\xi < 1$ the quantity $\xi p_{B_{n}}$ UTXO's is burned. When $\xi >1$ the quantity $\xi p_{B_{n}}$ UTXO's are minted (above the scheduled coinbase). 

Crucially, this method of proportional adjustment of the fee paid to the miner does not alter the Nash equilibrium of the pattern of fees offered by transactors. This is readily apparent by considering the incentive for $A$ or $B$ to alter their fee after a proportional adjustment $\xi$ of the miner's fee. Miners will still have an incentive to append transactions to blocks in ascending order of fees. If $B$ does not change its fee, $A$'s incentives are unaffected. It still must offer a higher fee than $B$ to secure a position in the queue of transactions ahead of $B$. Similarly, if $B$ was unwilling to offer a higher fee than $A$ before the adjustment, it not have an incentive to do so after the adjustment.

\subsubsection{Addition and Subtraction of UTXO Value}

When $p^{M}_{B_{n}}  - p_{B_{n}} > 0$, the difference is reflected in an additional payment to the miner above the total block reward. When $p^{M}_{B_{n}}  - p_{B_{n}} <  0$, the difference is reflected in a subtraction in payment to the miner from the total block reward. The source of increased UTXO value in the former case (minting or sending from another address) and the allocation of the unused UTXO value in the latter case, are dealt with by the \textit{Targeted} monetary policy in Section \ref{sec: Monetary Neutrality}.

\subsection{The puzzle difficulty signal and policy switch points}
\label{subsec: The mechanism}

Equation \ref{eq:Profit bound - static} shows that hashrate $N$ is over-determined. In addition to the miner block reward $P_{B}^{M}$, it is affected by mining cost $c$ and the exchange rate $e$. These variables move independently of each other. This means the block reward instrument will not be able to perfectly achieve the target. To compensate, $Target$ dynamically adjusts the block reward to move hashrate toward the target interval over time.

The mechanism works as follows. The endpoints of the hashrate target interval, $\{N_{UB}, N_{LB}\}$ map into puzzle difficulty signals $\{\frac{D}{T}(N_{UB}), \frac{D}{T}(N_{LB})\}$, which bound the puzzle difficulty target interval. An adjustment to the miner's block reward can be made after the end of each epoch when puzzle difficulty is adjusted. Epochs are numbered sequentially in temporal order, starting from 1. Epoch 1 includes blocks numbered 1 (Genesis) to $m$. Epoch 2 includes blocks numbered $m+1$ to $2m$, and so forth. The end-block of epoch $q$ is denoted $E_{q}$. The ratio $D/T$ is measured for epoch $q$ at block $E_{q}$ (the "adjustment block"). $D$ is the puzzle difficulty in effect for epoch $q$ and $N$ is the average hashrate in epoch $q$. It does not include the adjustment to $D$ made after $E_{q}$ is appended to the blockchain. The monotone relationship between the sufficient statistic and hashrate implies that hashrate is above target when $D/T > \frac{D}{T}(N_{UB})$ and below target when $D/T < \frac{D}{T}(N_{LB})$.

\textit{Targeted} imposes a ceiling on the miner's block reward in epoch $q+1$ when $D/T$, measured for epoch $q$, is above the target interval. After a ceiling has been put in place, it is lowered at each subsequent adjustment block until $D/T$ drops below $\frac{D}{T}(N_{UB})$. Similarly, \textit{Targeted} imposes a floor on the miner's block reward in epoch $q+1$ when $D/T$, measured for epoch $q$ when $D/T$ is below $\frac{D}{T}(N_{LB})$. When a floor has been put in place, it is raised at each subsequent adjustment block $E_{q}$ until the sufficient statistic rises above $\frac{D}{T}(N_{LB})$. When $D/T$ is inside the target interval, if a ceiling is in place from the prior adjustment block, the ceiling is raised, and so forth until it is removed altogether. If a floor is in place from the previous adjustment block, the floor is lowered at the current adjustment block, and so forth until it is removed altogether. The switching is displayed in Figure \ref{fig: Mining Puzzle Difficulty Signal Switch-Points}. At $E_{1}$ the sufficient statistic  crosses the $\frac{D}{T}(N^{UB})$ switch-point from below. This turns "on" the instrument that imposes a ceiling on the miner block reward, which is lowered at epoch $E_{2}$ and raised at $E_{3}$. $D/T$ is below $\frac{D}{T}(N_{LB})$ at $E_{4}$, which causes the ceiling to be removed and a floor to be imposed. The floor is raised at $E_{5}$.

\begin{figure}[h]
\centering
\includegraphics[page=1,width=0.75\textwidth,height = 0.3
\textheight]{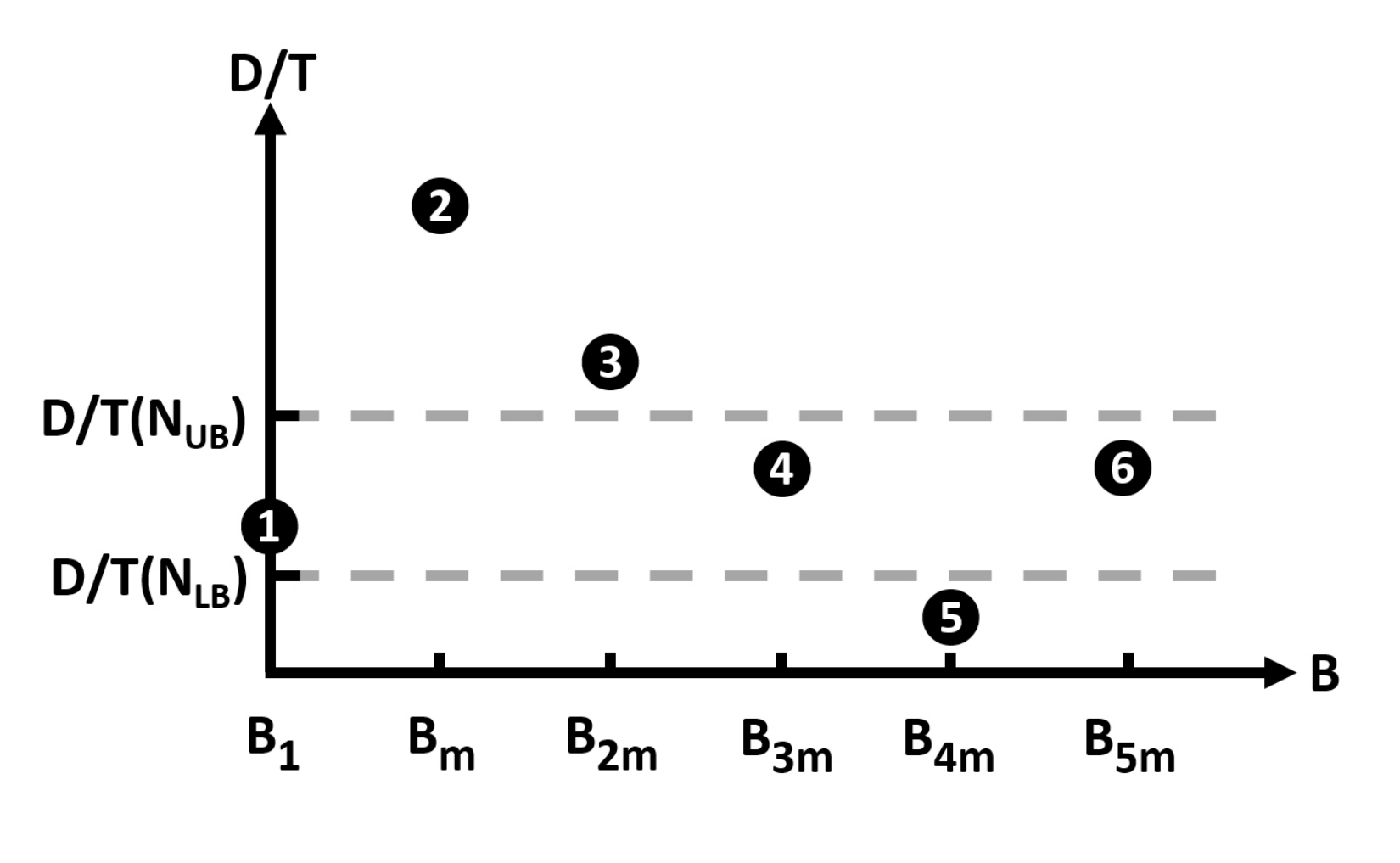}
\caption{Mining Puzzle Difficulty Signal Switch Points}\label{fig: Mining Puzzle Difficulty Signal Switch-Points}
\end{figure}
\vspace{-5pt}
The dynamic feature of the policy can be seen when another RHS variable of Equation \ref{eq: Link between the Signal and the Instrument} is increasing. For example, if the exchange rate is rising when  $D/T$ is above $\frac{D}{T}(N_{UB})$, the block reward ceiling will be pushed down at each adjustment point. It may transpire that hashrate continues to rise for some period of time, but so long as the exchange rate is bounded, the absolute value of the block reward ceiling adjustments will eventually rise above the exchange rate increase, which will push down hashrate.

\subsection{Policy to control hashrate}
\label{subsec: Policy to control hashrate}

The Hashrate control algorithm controls the imposition of ceilings and floors on the miner's block reward.\footnote{Appendix \ref{app:Code for The Hashrate Control Algorithm} contains Python code for the Hashrate Control Algorithm.}

\textbf{The Hashrate Control Algorithm}
\label{algo: The Hashrate Control Algorithm} 
A ceiling or floor on the miner's block reward is adjusted after the end-block $E_{q}$ of each epoch $q$  as a function of $D/T$. There are four cases to consider.

Set parameters $\tau, \gamma \in [0,1]$.

\underline{Case 1}: Hashrate is above the upper bound of the hashrate interval, $\frac{D}{T}(N_{UB})$. 

A ceiling on the miner's block reward is imposed. For each  block in epoch $q+1$, payment to miner is capped at $\text{ceil}_{q} = \tau \overline{P_{B}}(q), \; \tau < 1$, where $\overline{P_{B}}(q)$ is the median block reward in epoch $q$.

For each subsequent epoch until $\frac{D}{T}(N_{UB})$ is crossed from above, the ceiling is lowered by the formula $\text{ceil}_{q+1} \leftarrow \tau\text{ceil}_{q}$.

\underline{Case 2}: The puzzle adjustment at block $E_{q}$ implies that $D/T$ is below the lower bound of the target hashrate interval,  $\frac{D}{T}(N_{LB})$. 

A floor on the miner's block reward is imposed. For each  block in epoch $q+1$, payment to miner is capped at $\text{floor}_{q} = (1 + (1-\tau)) \overline{P_{B}}(q)$, where $\overline{P_{B}}(q)$ is the median block reward in Epoch $q$.

For each subsequent epoch until $\frac{D}{T}(N_{LB})$ is crossed from below, the floor is raised by the formula $\text{floor}_{q+1} \leftarrow (1 + (1-\tau))\text{ceil}_{q}$.

\underline{Case 3}: The puzzle adjustment at block $E_{q}$ implies that $D/T$ is inside the target hashrate interval. 

For each subsequent epoch until $\text{ceil}_{q} \leq \overline{P_{B}}(q)$, the ceiling is increased by the formula $\text{ceil}_{q+1} \leftarrow (1+(1 - \gamma)\text{ceil}_{q}$ Thereafter, the ceiling is removed.

For each subsequent epoch until $\text{floor}_{q} \geq \overline{P_{B}}(q)$, the floor is reduced by the formula $\text{floor}_{q+1} \leftarrow \gamma\text{floor}_{q}$ Thereafter, the floor is removed.\footnote{The formulas $\tau$ and $(1 +(1 - \tau))$ ensure symmetry in the adjustment of the ceiling and the floor, i.e. that the ceiling and floor adjust by the same percentage each epoch.}





\section{Strategic Effects of Block Adjustment Policy}
\label{sec:Strategic Effects of Block Adjustment Policy}

The adjustment of the miner's block reward introduces a new variable into equilibrium profit formula of Equation \ref{eq:Mining profit}. Treating the target hashrate as a point, $N^{T} = N_{UB} = N_{LB}$, the block reward can be expressed as $p^{M}_{B_{n}}(N^{T} - N)$. This shows that the block reward adjusts to deviations from the target hashrate. A question arises as to whether the block adjustment policy can induce strategic behavior that upsets the monotone response to block reward adjustments that was assumed in the previous section. Here we evaluate the stability of equilibrium at the targeted hashrate from two angles. First, we derive bounds on the block reward adjustment that ensure deviation from the equilibrium is unprofitable. Second, we show the adjustment mechanism causes hashrate to home in on the target under an assumption of constant miner profit. There may be other evaluations that could be carried out by placing different restrictions, but I believe the two chosen here show the reasonableness of the assumption that hashrate responds monotonically to changes in the block reward.\footnote{A maintained assumption of this paper is that block timestamps are correct. I do not explore whether \textit{Targeted} creates a new vulnerability to attack by miner manipulation of timestamps.}

\subsection{Stability of Miner Equilibrium}
\label{subsec:Stability of Miner Equilibrium}

Equation \ref{eq:Targeted miner profit} is the profit of a miner with hashrate $S$ at the target hashrate.

\begin{equation}
\label{eq:Targeted miner profit}
\pi(S) = \frac{S}{N^{T}}eP(N^{T})^{M}_{B}(1 - \Delta_{T})- cS
\end{equation}

The profit from deviating and adding or subtracting hashrate by $\Delta_{S}$ is

\[\pi(S + \Delta_{S}) = \frac{S + \Delta_{S}}{N^{T} + \Delta_{S}}eP(N^{T} + \Delta_{S})^{M}_{B}(1 - \Delta_{T})- c(S + \Delta_{S})\]

A deviation is unprofitable if $\pi(S + \Delta_{S}) < \pi(S),\; \forall \Delta_{S}$. Equation \ref{eq:Nash inequality} is a sufficient condition.

\begin{equation}
\label{eq:Nash inequality}
P(N^{T} + \Delta_{S}) < \frac{N^{T} + \Delta_{S}}{S + \Delta_{S}}\times\frac{S}{N^{T}}P(N^{T})^{M}_{B}
\end{equation}

Since $S \leq N^{T}$ by definition, it is immediate that, in order to maintain the inequality, an increase in hashrate, $\Delta_{S} > 0$ must cause the block reward to decrease and a decrease in hashrate does the opposite. This is consistent with the block adjustment policy. We can derive bounds on the magnitude of adjustment by evaluating the required adjustment at the extremes of $\Delta_{S} = 1$ and $\Delta_{S} =0$. This yields the range of adjustment multiple $(\frac{1 + \Delta_{S}/N^{T}}{1 + \Delta_{S}}, 1)$. Since the lower bond approaches 1 as $\Delta_{S}$ increases from 1 to $N^{T}$, the lower bound can be restricted to its value at $\Delta_{S} =1$, which yields $\frac{1 + 1/N^{T}}{2}$.\footnote{It is straightforward to add this constraint to the block adjustment formula replacing the point $N^{T}$ with the interval interval $N_{LB}, N_{UB}$.}

\subsection{Dynamic Adjustment of Hashrate}
\label{subsec: Dynamic Adjustment of Hashrate}

This section evaluates the dynamic response of hashrate to changes in the value of the miner's block reward which occur when hashrate is outside the target interval. The dynamic adjustment of hashrate to a change in the miner's block reward can be evaluated when hashrate adjusts to maintain a  profit $\phi$ per target time $T^{*}$.\footnote{The assumption of a constant profit  rules out the possibility of a dynamic strategy to influence hashrate by a miner coalition. Arguably, $\phi$  internalizes the exercise of market power.}

Consider an epoch $q$ where the total block reward is constant at $p^{M}_{B}$ and the puzzle difficulty adjustment at the end of epoch $q-1$ achieves the target growth rate when hashrate is $N$, and $N$ is the equilibrium hashrate when the miner's block reward is $p^{M}_{B}$. Suppose a ceiling or floor is imposed at the end of epoch $q-1$ resulting in a marginal change in the miner's block reward equal to $\Delta p$ (i.e. holding fixed the exchange rate, $e$, and hashrate cost $c$). Let $\Delta_{T}$ denote the distance between the actual blockchain growth rate and the target blockchain growth rate induced by a change in hashrate in epoch $q$, denoted $\Delta N$.

$\Delta_{T} =  T^{*} - T = \frac{D}{N} - \frac{D}{N + \Delta N}$

Equation \ref{eq: dynamic equilibrium profit} is the equilibrium of an interval of $T^{*}$ units of time in epoch $q$.

\begin{multline}
\label{eq: dynamic equilibrium profit}
\pi_{\Delta p,\Delta N} =\\
\underset{\text{lottery effect}}{\underbrace{\frac{1}{N + \Delta N}}}\Big (\underset{\text{growth rate effect}}{\underbrace{[1 + \underset{\Delta_{T}}{\underbrace{D\{\frac{1}{N} - \frac{1}{N + \Delta N}\}}}}}](e(P_{B}^{M} + \Delta p)\Big )\\
- c = \phi \footnote{Note that cost $c$ and profit $\phi$ are measured per unit of time and therefore scale with any deviation of blockchain growth rate from target.}
\end{multline}

The change in miner profit can be decomposed into a lottery effect and a growth rate effect. The lottery effect reflects that a change in hashrate changes the odds of any unit of hashrate mining the next block. The growth rate effect reflects that a change in hashrate alters the blockchain growth rate, which changes the number of puzzle lotteries in which a unit of hashrate can compete in $T^{*}$ units of time. Noting that $\frac{D}{N}|modulo\; T^{*} = 1$, the growth rate effect is $2 - \frac{D}{N + \Delta N}$. Equation \ref{eq: Marginal effect of hashrate on profit} is the marginal effect of hashrate on profit.

\begin{multline}
\label{eq: Marginal effect of hashrate on profit}
\frac{\Delta \pi}{\Delta N} = \big (e(P_{B}^{M} + \Delta p)\big ) \Big [ \frac{\Delta\text{lottery effect}}{\Delta N} + \frac{\Delta\text{growth rate effect}}{\Delta N}\Big]\\
= -2e(P_{B}^{M} + \Delta p)\frac{1}{(N + \Delta N)^{2}} < 0
\end{multline}

Equation \ref{eq: Marginal effect of miner's block reward on profit} is the marginal effect of the miner's block reward on profit.

\begin{equation}
\label{eq: Marginal effect of miner's block reward on profit}
\frac{\Delta \pi}{\Delta p} = \text{lottery effect}\times\text{growth rate effect}\times p^{M}_{B} > 0
\end{equation}
Using the implicit function theorem to combine equations \ref{eq: Marginal effect of hashrate on profit} and \ref{eq: Marginal effect of miner's block reward on profit}  shows that an increase in the block reward induces an increase in hashrate (and vice versa) when hashrate adjusts to maintain a constant profit rate per unit of time
$\Delta N/\Delta p = -\frac{\Delta \pi}{\Delta p}\big /\frac{\Delta \pi}{\Delta N} = +/+ > 0$.\footnote{The analysis of dynamic equilibrium can be extended to the case where the difficulty adjustment at the end of epoch $q$ does not achieve the target growth rate. In that case the ratio $D/N \neq T^{*}$ and 
\[\frac{\Delta \pi}{\Delta N} = -e(P_{B}^{M} + \Delta p)\frac{1}{(N + \Delta N)^{2}}\big (1 + D/N\big ) < 0\]
which yields the same result in terms of the direction of the response of hashrate to the value of the miner's block reward.} The key point is that, when hashrate is responsive to changes in the block reward, the marginal effect of the Hashrate Control Algorithm is to induce an adjustment of hashrate in the intended direction when a ceiling or floor pushes the miner's block reward below or above the total block reward.\footnote{Assuming that $e$ and $c$ are bounded, the result can easily be extended to show that the Hashrate Control Algorithm will eventually push hashrate toward the target interval, but it will not necessarily do so monotonically.} The foregoing proves the following proposition.\footnote{Bissias et.al. (2022) \cite{Bissias} document that miner spending allocation between Bitcoin and Bitcoin Cash adjust at high frequency to maintain a zero arbitrage condition. This validates that hashrate tends to adjust quickly. A richer formulation of my model would include the shadow profit of mining an alternative blockchain.} 

\begin{proposition*}
The marginal effect of the Hashrate Control Algorithm is to dynamically push hashrate toward the target interval.
\end{proposition*}

\section{Monetary Neutrality}
\label{sec: Monetary Neutrality}

This section describes a protocol that ensures the implementation of the Hashrate Control Algorithm does not cause any substantive change to the Nakamoto monetary policy in terms of the spending potential of addresses. \textit{Targeted} induces offsetting additions or subtractions from the UTXO value of an address that occurs when the address sends its UTXO value in a transaction. The concept of money supply is adapted to include the  UTXO adjustments that are committed for future transactions.  

\begin{definition*}{Spending Potential}

(i) \textbf{Address Spending Potential} The UTXO value in address $x$  plus the UTXO value that is scheduled to be added to or subtracted from  $x$ when it sends a transaction. 

(ii) \textbf{Aggregate Spending Potential} The sum of the spending potential of addresses plus the increase or decrease in spending potential that is not currently assigned to addresses, but which under the protocol can be assigned in the future.
\end{definition*}

The concept of monetary neutrality has two dimensions.

\begin{definition**}{Monetary Neutrality}

(i) \textbf{Aggregate Monetary Neutrality} Aggregate UTXO spending  potential under Targeted is the same as Nakamoto at every block. 

(ii) \textbf{Distributional Monetary Neutrality} At any block, an addition or subtraction of UTXO from the total block reward is allocated to an address in proportion to its percentage of aggregate spending potential. 
\end{definition**}

\subsection{The \textit{Targeted} Monetary Policy}
\label{subsec: The Targeted Monetary Policy}

This section presents a high-level description of the monetary policy as it affects a single address $x$. Python code describing the logic of the policy is in Appendix \ref{app: Code for the Targeted Monetary Policy}.

\textbf{Subtractions from the total block reward}: When a UTXO value of $q$ is subtracted from the total block reward at block $B_{n}$ the miner is instructed to send $q$ to an address denoted the "UTXO pool". The UTXO pool is a single address from which UTXO value is sent to other addresses and has an initial value $P$ to which $q$ is added. Every address $x$ holding UTXO value is assigned spending potential in the UTXO pool (the "assigned value") which is denoted $e_{x}$. When $q$ is sent to the UTXO pool, $x$'s assigned value is increased by $x$'s fraction of aggregate UTXO spending potential implied by the state after the block $B_{n}$ transactions, denoted $S_{x|B_{n-1}}$ ($x$'s  "percentage") times $q$, rounded down to the nearest Satoshi.

\textbf{Additions to the total block reward}: When  $q$ is added to  the total block reward at block $B_{n}$ there are two possible adjustments to accounts. If $P - q > 0$, $x$'s assigned value in the UTXO pool is reduced by the product of its percentage times $q$, rounded down to the nearest Satoshi. When $P - q \leq 0$ two accounting changes occur; (i)  P and $x$'s assigned value $e_{x}$ are each reduced to zero  and (ii) the value $q - P$ is added to an account that tracks value that is subtracted from the spending potential of addresses (the "Remainder Accounts") which has initial value $R$. Each address $x$ is assigned a value in the remainder account, denoted $r_{x}$, which is increased by  its percentage times $P -q$, rounded down to the nearest Satoshi.\footnote{The UTXO value subtracted from P can either be burned or sent to the miner. The latter would reduce the minting of UTXO required to increase the miner's block reward.}

Note that the spending potential of an address $x$ with UTXO value $x$ is $x + e_{x} - r_{x}$.\vskip5pt

\textbf{Rounding Adjustment} At each block $B_{n}$ the unassigned UTXO value in the UTXO pool and the Remainder Accounts are assigned to the individual address accounts, $e_{x}$ and $r_{x}$ in accordance to the percentage of address $x$, rounded down to the nearest Satoshi.

\textbf{Transactions} Prior to the adjustments to the UTXO pool and Remainder pool, the transactions in block $B_{n}$ are executed. An adjustment is made when address $x$ sends UTXO value in a transaction that is appended to block $B_{n}$. If $e_{x} - r_{x} >0$, the difference is sent to $x$. If $r_{x} - e_{x} > 0$, the difference is burned from $x$'s UTXO value.\footnote{The burning requirement invalidates the transaction when the burn value exceeds the unspent UTXO.} The initial values of $e_{x}$ and $r_{x}$ are the ending values for block $B_{n-1}$. The output values of $S_{x|B_{n,1}}, e_{x}$ and $r_{x}$ are the initial state variables in the UTXO pool and the Remainder Accounts adjustments. 

Figure \ref{fig: Targeted Monetary Policy} depicts the operation of the \textit{Targeted} monetary policy at block $B_{n}$. The timing sequence is denoted by the circled numbers. The first step is the execution of the transactions. The adjustment to a transaction sender's UTXO value is determined by the difference between its account value in the UTXO pool and its account value in the Remainder Accounts. The second step is the payment of the miner block reward. If $P^{M}_{B_{n}} < p^{M}_{B_{n}}$, UTXO value of $q$ is minted to cover the payment. The third step is to send of UTXO value $q$ to the UTXO Pool if $P^{M}_{B_{n}} < p^{M}_{B_{n}}$ or else burn UTXO value up to $q$ in the UTXO pool and increase the value of the Remainder Accounts by the residual. Adjustments to the address account values in the UTXO Pool and the Remainder Accounts are made.\footnote{ Steps 2 and 3 can be simultaneous. When $q < 0$ the minting of UTXO value is reduced by the amount of value in the UTXO pool, which is sent to the miner from the UTXO pool. This is implicit in the code in Appendix \ref{app: Code for the Targeted Monetary Policy}.}

\begin{figure}[H]
\centering
\includegraphics[page=1,width=0.95\textwidth,height = 0.45
\textheight]{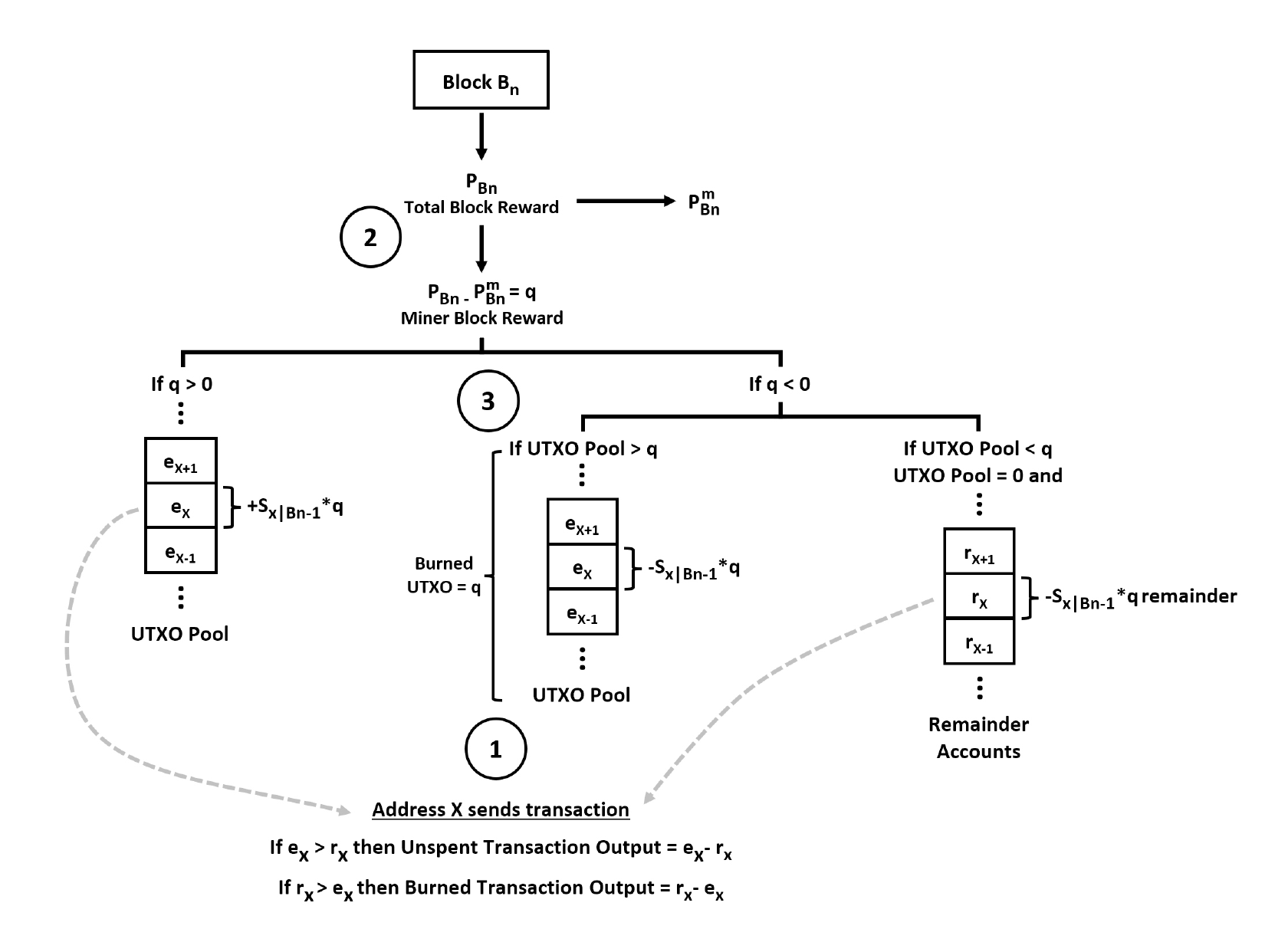}
\caption{Targeted Monetary Policy}
\label{fig: Targeted Monetary Policy}
\end{figure}

\begin{proposition**}
Targeted Achieves Monetary Neutrality
\end{proposition**}

\begin{proof}
The difference between the total block reward and the miner's block reward, $q$, is offset by changes of opposite sign in the value of the UTXO pool and the Remainder Accounts; $q = \Delta P + \Delta R$. This proves aggregate monetary neutrality. The value change in the UTXO Pool and the Remainder Accounts are allocated to addresses in proportion to their spending potential, rounded down to the nearest Satoshi and the unassigned residual is assigned to addresses in accordance with the same formula at each block. This proves distributional neutrality.  
\end{proof}

\section{Conclusion}

\textit{Targeted} is motivated by two observations about PoW blockchains. One observation is the two external hashrate costs of network security and $CO^{2}$ emissions, the former decreasing and the latter increasing in hashrate. This observation suggests there is an interval of hashrate in which total costs are minimized. The second observation is that hashrate is a function of the value of the block reward. Nakamoto does not provide an incentive for miners to set hashrate at the optimal level. \textit{Targeted} incentivizes miners to  home in on a target hashrate interval. This is achieved by imposing a floor to the miner's block reward when hashrate is below target and imposing a ceiling on the  miner's block reward when hashrate is above target. The ceiling is progressively lowered and the floor is progressively raised for so long as hashrate is outside the target interval. \textit{Targeted} achieves aggregate and distributional monetary neutrality. Aggregate monetary neutrality means that the UTXO spending potential of the network is the same under \textit{Targeted} and Nakamoto at each block. Distributional monetary neutrality means that the UTXO spending potential of addresses are adjusted to offset increases and decreases in the block reward in proportion to their UTXO value. 

Identifying the optimal hashrate interval entails some difficulty. Three of the key inputs; estimating hashrate, ASIC efficiency and electricity cost, are, in principle, solvable. But quantifying the other two inputs; the costs of $CO^{2}$ emissions and network security, are elusive. To implement \textit{Targeted} at the present time will require the formation of a consensus among network participants concerning those quantifications.  



\bibliographystyle{plain}
\bibliography{mybibliography}

\pagebreak 

\appendix

\section{Code for The Hashrate Control Algorithm}
\label{app:Code for The Hashrate Control Algorithm}

This appendix states the logic steps of the Hashrate Control Algorithm. It is not the code that would be written to implement the \textit{Targeted} on servers in the Bitcoin network.
{\fontsize{9}{9}\selectfont
\begin{lstlisting}[language=Python, caption=Python code for Hashrate Control Algorithm]


class Epoch:
    def __init__(self):
        self.rewards = []  # To store the rewards of each block
        within an epoch
        self.hashrates = [] # To store the hashrates of each block
        within an epoch
        self.ceil = None
        self.floor = None
        self.difficulty = None
        self.elapsed_time = None # T_n

    def median_block_reward(self):
        sorted_rewards = sorted(self.rewards)
        return np.median(sorted_rewards) 

    def average_hashrate(self):
        sorted_hashrates = sorted(self.hashrates) 
        return np.mean(sorted_hashrates)

    def add_reward(self, reward):
        self.rewards.append(reward)

    def add_hashrate(self, n):
        self.hashrates.append(n)

class Blockchain:
    def __init__(self, tau, gamma, upper_bound, lower_bound, initial_difficulty):
        self.tau = tau
        self.gamma = gamma
        self.epochs = [Epoch()]
        self.DT = None  # Puzzle difficulty/Blockchain growth rate;
        get this value from blockchain source code
        self.DT_N_UB = upper_bound  # Upper boundary
        self.DT_N_LB = lower_bound  # Lower boundary

    def adjust_reward(self, reward, epoch_idx):
        epoch = self.epochs[epoch_idx]
        if epoch.ceil is not None:
            reward = min(reward, epoch.ceil)
        if epoch.floor is not None:
            reward = max(reward, epoch.floor)
        epoch.add_reward(reward)

    def adjust_hashrate(self, hashrate, epoch_idx, e_current, P_current, efficiency_current, electricity_cost_current):
        alpha1 = None
        alpha2 = None
        alpha3 = None 
        alpha4 = None 
        # Found through regression on past values

        epoch = self.epochs[epoch_idx]
        log_eP_current = np.log(1+ (e_current * P_current)) 
        log_eff_current = np.log(1+efficiency_current)
        log_electricity_cost_current = np.log(1+electricity_cost_current)
        # Put it all together:
        # bN_{t+1} = alpha1 + alpha2*log(eP_t) + + alpha3*log(efficiency_t) + alpha4*log(c_t)
        predicted_log_hashrate = (
            alpha1 
            + alpha2 * log_eP_current
            + alpha3 * log_eff_current
            + alpha4 * log_electricity_cost_current
        )

        # Exponentiate to get N_{t+1}
        new_hashrate = np.exp(predicted_log_hashrate)      
        epoch.add_hashrate(new_hashrate)
        return new_hashrate
        
    def end_of_epoch(self):
        last_epoch = self.epochs[-1]
        P_B_median = last_epoch.median_block_reward()
        N_n = last_epoch.average_hashrate()
        new_epoch = Epoch()

        BLOCKS_PER_EPOCH = 14
        T_STAR = BLOCKS_PER_EPOCH * 10 * 60 # 1,209,600 for 2016 blocks
        D_n = last_epoch.difficulty
    
        T_n = (BLOCKS_PER_EPOCH * D_n * 2**32) / N_n
        last_epoch.elapsed_time = T_n

        T_hat = max(T_STAR /4, min(T_n, 4 *T_STAR))
        D_next = D_n * (T_STAR / T_hat)
        new_epoch.difficulty = D_next

        # Case 1: Hashrate too high - need to decrease rewards
        if self.DT > self.DT_N_UB:
            #("Case 1: Hashrate above upper bound")
            new_epoch.ceil = (self.tau) * P_B_median
            new_epoch.floor = None  # Clear any existing floor
    
        # Case 2: Hashrate within bounds - gradual adjustment
        elif self.DT_N_LB < self.DT < self.DT_N_UB:
            
            # Previous epoch had a ceiling
            if last_epoch.ceil is not None:
                if last_epoch.ceil >= self.DT:
                    # Gradually relax the ceiling
                    new_epoch.ceil =  (1+self.gamma) * last_epoch.ceil
                else:
                    # Remove ceiling if it's no longer needed
                    new_epoch.ceil = None
            
            # Previous epoch had a floor
            elif last_epoch.floor is not None:
                if last_epoch.floor <= self.DT:
                    # Gradually relax the floor
                    new_epoch.ceil = (self.gamma) * last_epoch.floor
                else:
                    # Remove floor if it's no longer needed
                    new_epoch.floor = None
            else:
                new_epoch.floor = last_epoch.floor
                new_epoch.ceil = last_epoch.ceil

        # Case 3: Hashrate too low - need to increase rewards
        elif self.DT < self.DT_N_LB:
            #("Case 3: Hashrate below lower bound")
            new_epoch.floor =  (1+(1-self.tau)) * P_B_median
            new_epoch.ceil = None  # Clear any existing ceiling

        self.epochs.append(new_epoch)
\end{lstlisting}
}

\section{Code for the \textit{Targeted} Monetary Policy}
\label{app: Code for the Targeted Monetary Policy}

This appendix states Python code for the \textit{Targeted} monetary policy. The interpretation of the variables are commented. The code states the logic for the re-allocation of spending potential of addresses at each block that is induced by transactions and the addition or subtraction of the block reward, $p_{B_{n}}^{m} - p_{B_{n}}$, in accordance with the Hashrate Control Algorithm. It is not the code that would be written to implement the \textit{Targeted} on servers in the Bitcoin network.

{\fontsize{9}{9}\selectfont
\begin{lstlisting}[language=Python, caption= Python code for \textit{Targeted} monetary policy]

def update_utxo_values(p_B_n_m, p_B_n, S_x_B_n_minus_1, P, FP, e_values, r_values, R, FR):
    q = p_B_n_m - p_B_n
    
    for x in range(len(e_values)):
        if q < 0:
            P += q
            e_values[x] = math.floor(e_values[x] + S_x_B_n_minus_1[x] * (q + FP))
        elif q > 0:
            if P > q:
                P -= q
                e_values[x] = math.floor(e_values[x] - S_x_B_n_minus_1[x] * (q - FP))
            else:
                R += q - P
                r_values[x] = math.floor(r_values[x] + S_x_B_n_minus_1[x] * (P + FR))
                P = 0
                e_values[x] = 0

    x_values = [x_initial[i] + e_values[i] - r_values[i] for i in range(len(e_values))]
    return P, e_values, r_values, R, FP, FR, x_values

def process_transactions(x_initial, e_values, r_values, transaction_indices):

# initial spending potential of addresses
x_values = [x_initial[i] + e_values[i] - r_values[i] for i in range(len(e_values))]

    #transaction sender update step
    for i in transaction_indices:
        x_value = math.floor(x_initial[i] + e_values[i] - r_values[i])
        e_values[i] = 0
        r_values[i] = 0

     #miner update step
     x_value[miner_index] += p_B_n_m

    total_sum = sum(x_values[i] + e_values[i] - r_values[i] for i in range(len(e_values)))

     # Spending potential shares of addresses after transactions and miner fee, before allocations to UTXO Pool and Remainder Accounts
     S_x_B_n_minus_1 = [(x_values[i] + e_values[i] - r_values[i]) / total_sum for i in range(len(e_values))]

    return e_values, r_values, S_x_B_n_minus_1

def update_fp_fr(FP, FR, e_values, r_values):
    FP += sum(e_values)
    FR += sum(r_values)
    return FP, FR

def compute_S_x_B_n(x_values, e_values, r_values):
    total_sum = sum(x + e - r for x, e, r in zip(x_values, e_values, r_values))

     # Spending potential shares of addresses after transactions and miner fee and assingments to UTXO Pool and Remainder Accounts. 
     S_x_B_n = [(x + e - r) / total_sum for x, e, r in zip(x_values, e_values, r_values)]
    return S_x_B_n

Interpretation of variables:

x_initial = [<values to be inserted>]  # initial values for x 
e_initial = [<values to be inserted>]  # initial UTXO values of addresses
r_initial = [<values to be inserted>]  # initial values of r_x

p_B_n_m = <value to be inserted> # miner's block reward
p_B_n = <value to be inserted> # total block reward
P = <value to be inserted> # aggregate value of UTXO Pool 
e_values = e_initial[:]  # copy of initial UTXO values in UTXO Pool
r_values = r_initial[:]  # copy of initial values of r_x, the value for each address in the Remainder Accounts
R = <value to be inserted> #  aggregate value in the Remainder Accounts
FP = <value to be inserted> #  sum of values e_x in the UTXO Pool
FR = <values to be inserted> # sum of values r_x in the Remainder Accounts
transaction_indices = [:] # indices of addresses that send transactions in block B_n
miner_index = x # address of the miner of block B_N

Computation:

# Process transactions for each address in transaction_indices
e_values, r_values, S_x_B_n_minus_1 = process_transactions(x_initial, e_values, r_values, transaction_indices, p_B_n_m, miner_index)

# Update UTXO values for each address
P, e_values, r_values, R, FP, FR, x_values = update_utxo_values(p_B_n_m, p_B_n, S_x_B_n_minus_1, P, FP, e_values, r_values, R, FR)

# Update FP and FR with the residuals
FP, FR = update_fp_fr(FP, FR, e_values, r_values)

# Compute S_x_B_n for each address
S_x_B_n = compute_S_x_B_n(x_values, e_values, r_values)

# Compute the sum of initial and final x values
initial_x_sum = sum(x_initial)
final_x_sum = sum(x_values)
x_difference = initial_x_sum - final_x_sum

print("Updated values:")
print(f"P: {P}, e_values: {e_values}, r_values: {r_values}, x_values: {x_values}, R: {R}, FP: {FP}, FR: {FR}")
print(f"S_x_B_n: {S_x_B_n}")
print(f"Difference between the sum of initial and final x values: {x_difference}")

\end{lstlisting}
}

\end{document}